\documentclass[12pt]{article}
\usepackage{graphicx, amsmath, amssymb, bm,enumerate,comment,amsthm}
\usepackage{fancybox, boxedminipage}
\usepackage{ascmac}
\usepackage{fancyheadings}
\usepackage{fancybox}
\usepackage[english]{babel}
\usepackage{authblk}

\newtheorem{theorem}{Theorem}
\newtheorem{lemma}{Lemma}

\newtheorem{proposition}{Proposition}

\def \R {\mathbb{R}}
\def \Z {\mathbb{Z}}

\begin{document}

\title{A 4-Approximation Algorithm for $k$-Prize Collecting Steiner Tree Problems
}


\author{Yusa Matsuda\; and \;Satoshi Takahashi\\ 
        The University of Electro-Communications, Japan
}

\date{February 19, 2018}

\maketitle

\begin{abstract}
This paper studies a 4-approximation algorithm for $k$-prize collecting Steiner tree problems. This problem generalizes both $k$-minimum spanning tree problems and prize collecting Steiner tree problems. Our proposed algorithm employs two 2-approximation algorithms for $k$-minimum spanning tree problems and prize collecting Steiner tree problems. Also our algorithm framework can be applied to a special case of $k$-prize collecting traveling salesman problems.

\end{abstract}

\section{Introduction}
In the network design problems, we have been studying many combinatorial problems such as minimum Steiner tree problems and minimum spanning tree problems\cite{Korte}. 
This paper studies a $k$-prize collecting Steiner tree problems ($k$-PCST) which is proposed by Han et.al.\cite{Han2016}. This problem generalizes both $k$-minimum spanning tree problems and prize collecting Steiner tree problems. 

We are given an undirected connected graph $G=(V,E)$, a root vertex $r\in V$, and an integer $k$. Let $c:E\to \R_+$ be a nonnegative cost function on $E$ and $\pi:V\to \R_+$ be a nonnegative penalty cost function on $V$. Suppose that a tuple $(G, r, k, c, \pi)$ is an instance of $k$-PCST. The goal is to find an $r$-rooted subtree $F=(V_F, E_F)$ that spans at least $k$ vertices with minimum total cost that is described by $\sum_{e\in E_F}c(e)+\sum_{v\in V\setminus V_F}\pi(v)$. Han et.al.\cite{Han2016} proposed a 5-approximation algorithm based on primal-dual method for this problem. Suppose OPT and $F_{k-{\rm PCST}}^{*}=(V_{k-{\rm PCST}}^{*}, E_{k-{\rm PCST}}^{*})$ are the optimal value and an optimal solution for the $k$-PCST instance.

Han et. al. \cite{Han2016} assume the following four general conditions to establish 5-approximation algorithm for $k$-PCST.  
\begin{itemize}
\item $G$ is a complete graph.
\item The cost function $c$ satisfies a triangular inequality.
\item For each $v\in V$, a distance from $r$ is smaller than the optimal value of PCST obtained from the $k$-PCST instance.
\item The minimum nonnegative value of cost function is smaller than the optimal value of $k$-PCST.
\end{itemize}
Our paper employs same conditions when we discuss $k$-PCST.

In this paper, we propose a 4-approximation algorithm for $k$-PCST based on two primal-dual methods for $k$-minimum spanning tree problems and prize collecting Steiner tree problems. Also we propose 4-approximation algorithm, that is the same framework, for $k$-prize collecting traveling salesman problems under the triangular inequality.

\section{Preliminaries}
The $k$-prize collecting Steiner tree problem is a generalization problem of both $k$-minimum spanning tree problems ($k$-MST) and prize collecting Steiner tree problems (PCST)\cite{Han2016}. 

$k$-MST is one of well studied network design problems. 
Given an undirected connected graph $G=(V,E)$, a vertex $r\in V$, and an integer $k$. Let $c:E\to \R_+$ be a nonnegative cost function on $E$. Suppose that a tuple $(G, r, k, c)$ is an instance of rooted $k$-MST.
The goal of this problem is to find a subtree $F = (V_F, E_F)$ of $G$ such that $F$ includes the vertex $r$, $|V_F|\geq k$, and minimum total cost which is $\sum_{e \in E_F}c(e)$.
Unrooted $k$-MST was proposed by \cite{Ravi1996} who proved its NP-hardness and gave a $3\sqrt{k}$-approximation algorithm. After that, many approximation algorithms have been proposed (see \cite{Fischetti1994}, \cite{Garg1996}, \cite{Blum1999} and \cite{Arora2006}). The best approximation ratio for this problem is 2 by Garg\cite{Garg2005} who treats rooted $k$-MST.

For the $k$-PCST instance $(G, r,k,c,\pi)$, we can get a rooted $k$-MST instance by $\pi(v)=0$ for each $v\in V$. Suppose that OPT$_{k\text{-MST}}$ and $F_{k\text{-MST}}^{*}=(V_{k\text{-MST}}^{*}, E_{k\text{-MST}}^{*})$ are the optimal value and an optimal solution for the $k$-MST instance induced by the $k$-PCST instance.
We hold the following lemma.
\begin{lemma}\label{lm:lm1}
An optimal solution $F^{*}_{k\text{\rm -MST}}$ for $k${\rm -MST} satisfies
\[
\sum_{e \in E_{k\text{\rm -MST}}^{*}}c(e) \leq \mbox{\rm OPT}.
\]
\end{lemma}
\begin{proof}
Since the cost function $c$ and the penalty function $\pi$ are nonnegative, and $|V_{k\text{-PCST}}^{*}| \geq k$, 
\begin{eqnarray*}
\mathrm{OPT} &=& \sum_{c \in E_{k\text{-PCST}}^{*}}c(e) + \sum_{v \not\in V_{k\text{-PCST}}^{*}}\pi(v) \\
& \geq & \sum_{e \in E_{k\text{-PCST}}^{*}}c(e) \\
& \geq & \sum_{e \in E_{k\text{-MST}}^{*}}c(e).
\end{eqnarray*}
\end{proof} 
From the lemma \ref{lm:lm1}, a set of feasible solutions of $k$-MST includes any feasible solutions of $k$-PCST.

Other problem is prize collecting Steiner tree problems (PCST). In the PCST, given an undirected connected graph $G=(V,E)$, an root vertex $r\in V$. Let $c:E\to \R_+$ be a nonnegative cost function on $E$ and $\pi:V\to \R_+$ be a nonnegative penalty cost function on $V$. The goal is to find an $r$-rooted subtree $F=(V_F, E_F)$ with minimum total cost that is $\sum_{e\in E_V}c(e)+\sum_{v\in V\setminus V_F}\pi(v)$. Goemans and Williamson derive a 2-approximation algorithm for this problem\cite{G-W}. They show the following theorem.
\begin{theorem}[\cite{G-W}]\label{thm:thm1}
Consider an arbitrary PCST instance where $|V|\ge 2$. Suppose that $F=(V_F, E_F)$ and $\bm{y}$ is an output of the algorithm. It holds that
\[
\sum_{e\in E_F}c(e)+\left(2-\frac{1}{n-1}\right)\sum_{v\notin V_F}\pi(v)\le \left(2-\frac{1}{n-1}\right)\sum_{S\subseteq V\setminus\{r\}}y_S.
\]
\end{theorem}
$\sum_{S\subseteq V\setminus\{r\}}y_S$ is the objective function of dual problem of the LP relaxed PCST.
This theorem shows that as the number of vertices increases, the approximation ratio becomes asymptotically 2.
The best possible approximation algorithm is 1.9672 approximation derived by Archer et. al.\cite{Archer}. When we set $k=0$, PCST is a special case of $k$-PSCT. Since every feasible solution of $k$-PCST is also a feasible solution of PCST, we are able to show the following lemma.
\begin{lemma}\label{lm:lm2}
An optimal solution $F^{*}_{\rm PCST}=(V_{\rm PCST}^{*}, E_{\rm PCST}^{*})$ for {\rm PCST} satisfies
\[
\sum_{e \in E_{\rm PCST}^{*}}c(e)+\sum_{v \notin V_{\rm PCST}^{*}}\pi(v)\leq \mbox{\rm OPT}.
\]
\end{lemma}

\section{Proposed algorithm}
We propose a primal dual method based 4-approximation algorithm for $k$-PCST. 
Our algorithm combines two procedures, one is 2-approximation algorithm for PCST proposed by Goemans and Williamson\cite{G-W}, the other is 2-approximation algorithm for $k$-MST proposed by Garg\cite{Garg2005}. We call each approximation algorithm procedure 1 and procedure 2, respectively.
\begin{itembox}[l]{4-approximation algorithm for $k$-PCST}
\begin{description}
\item [Input: ] An undirected connected graph $G = (V, E)$, an root vertex $r \in V$, an integer $k$, a cost function $c: E \to \mathbb{R}_+$, and a penalty function $\pi: V\to \R_+$.
\item [Output: ] $r$-rooted $k$-prize collecting Steiner tree of $G$.
\item [Step 1: ]For a PCST instance $(G, r, c, \pi)$, apply Procedure 1. Let $F_{\rm PCST}=(V_{\rm PCST}, E_{\rm PCST})$ be an output of this procedure. If $|V_{\rm PCST}| \geq k$, return $F_{\rm PCST}$, otherwise goto Step 2.
\item [Step 2: ]For a $k$-MST instance $(G, r, k,c)$, apply Procedure 2. Let $F_{k\text{-MST}}=(V_{k\text{-MST}}, E_{k\text{-MST}})$ be an output of this procedure.
\item [Step 3: ] Construct an undirected connected graph $G' = (V_{\rm PCST} \cup V_{k\text{-MST}}$, $E_{\rm PCST} \cup E_{k\text{-MST}})$, and find a minimum spanning tree $F_{\rm OUT} = (V_{\rm OUT}, E_{\rm OUT})$ of $G'$. Return $F_{\rm OUT}$.
\end{description}
\end{itembox}

It is easy to check the output of proposed algorithm is a feasible solution for $k$-PCST and the algorithm is an polynomial time algorithm. The analysis of the time complexity will be described later. To analyze an approximation ratio of our algorithm, we introduce some terms. Given a $k$-PCST instance $(G, r, k, c, \pi)$, let $F_{k\text{-PCST}}^*$ be an optimal solution and OPT be its optimal value. Also let $F_{\rm PCST}^*$ be an optimal solution of PCST and OPT$_{\rm PCST}$ be its optimal value, $F_{k\text{-MST}}^*$ be an optimal solution of $k$-MST and OPT$_{k\text{-MST}}$ be its optimal value.

\begin{lemma}\label{lm:lm3}
If proposed algorithm terminated at Step 1, then $F_{\rm OUT} = (V_{\rm OUT}, E_{\rm OUT})$ holds 
\[
\sum_{e \in E_{\rm OUT}}c(e) + \sum_{v \not\in V_{\rm OUT}}\pi(v) \leq 2\mathrm{OPT}.
\]
\end{lemma}
\begin{proof}
From the assumption, the output $F_{\rm OUT}$ is obtainable from Procedure 1. Since $F_{\rm OUT}$ contains a vertex $r$ and $|V_{\rm OUT}|\ge k$, $F_{\rm OUT}$ is a feasible solution of $k$-PCST. Suppose $\mathrm{OPT}_{\rm PCST}$ is an optimal value of PCST instance $(G, r, c, \pi)$. From lemma \ref{lm:lm2}, we get
\begin{eqnarray*}
\sum_{e \in E_{\rm OUT}}c(e) + \sum_{v \not\in V_{\rm OUT}}\pi(v) &\leq& 2\mathrm{OPT}_{\rm PCST}\\ &\leq& 2\mathrm{OPT}.
\end{eqnarray*}
\end{proof}

\begin{theorem}\label{thm:thm1}
 If proposed algorithm terminated at Step 4, then $F_{\rm OUT} = (V_{\rm OUT}, E_{\rm OUT})$ holds 
 \[
 \sum_{e \in E_{\rm OUT}}c(e) + \sum_{v \not\in V_{\rm OUT}}\pi(v) \leq 4\mathrm{OPT}.
 \]
\end{theorem}

\begin{proof}
Since $F_{\rm OUT}$ is a spanning tree of $G'$, it holds $V_{\rm OUT} = V_{\rm PCST} \cup V_{k\text{-MST}}$ and $E_{\rm OUT} \subseteq E_{\rm PCST} \cup E_{k\text{-MST}}$. Since $r \in V_{\rm PCST}$ and $|V_{k\text{-MST}}| = k$ satisfy, $F_{\rm OUT}$ is a feasible solution of $k$-PCST. From the lemma \ref{lm:lm1} \ref{lm:lm2}, and \ref{lm:lm3}, we get 
\begin{eqnarray*}
\sum_{e \in E_{\rm OUT}}c(e) + \sum_{v \not\in V_{\rm OUT}}\pi(v) &\leq& \sum_{e \in E_{k\text{-MST}}}c(e) + \sum_{e \in E_{\rm PCST}}c(e) + \sum_{v \not\in V_{\rm PCST}}\pi(v)\\
& \leq & 2\mathrm{OPT}_{k\text{-MST}} + 2\mathrm{OPT}_{\rm PCST}\\
& \leq & 2\mathrm{OPT} + 2\mathrm{OPT}\\ &=& 4\mathrm{OPT}.
\end{eqnarray*}
\end{proof} 

From above discussion, we get 4-approximation algorithm for $k$-PCST. Our algorithm has a following property.

\begin{proposition}
If there exists an optimal solution $F^{\star} = (V^{\star}, E^{\star})$ for $k$-{\rm PCST} such that $\sum_{e \in E^{\star}}c(e) \leq \frac{\mathrm{OPT}}{2}$, then $F_{\rm OUT} = (V_{\rm OUT}, E_{\rm OUT})$ holds 
\[\sum_{e \in E_{\rm OUT}}c(e) + \sum_{v \not\in V_{\rm OUT}}\pi(v) \leq 3\mathrm{OPT}.
\]
\end{proposition}

\begin{proof}
Consider the case that our algorithm terminated in Step 4. From the theorem \ref{thm:thm1}, $F_{\rm OUT}$ is a feasible solution of $k$-PCST. We get
\begin{eqnarray*}
\sum_{e \in E_{\rm OUT}}c(e) + \sum_{v \not\in V_{\rm OUT}}\pi(v) &\leq& 2\mathrm{OPT}_{k\text{-MST}} + 2\mathrm{OPT}_{\rm PCST} \\
	&\leq& 2\sum_{e \in E^{\star}}c(e) + 2\mathrm{OPT}_{\rm PCST} \\
    &\leq& \mathrm{OPT} + 2\mathrm{OPT}\\ &=& 3\mathrm{OPT}.
\end{eqnarray*}
\end{proof} 
Next, we describe the time complexity of the proposed algorithm for the $k$-PCST problem. If the number of vertices is $n$ and the number of edges is $m$ for the input graph, the time complexity of GW-algorithm is O($n^2 \log n$), the time complexity of Garg's algorithm is O($mn^4\log n$). 
In the proposed algorithm, we call the algorithm of GW-algorithm and Garg's algorithm once. The minimum spanning tree is obtained for the subgraph obtained last, but by using the Kruskal algorithm O($m\log n$). Therefore, the time complexity of the proposed algorithm is O($mn^4\log n$).

\section{$k$-Prize collecting traveling salesman problems}
In order to simplify the explanation, a tour of the graph $G$ is represented as a sequence of vertices, denoted by $T=(u_1,u_2,...,u_k,u_1)$. Denote that $E[T]$ is a edge set and $V[T]$ is a vertex set induced by $T$. In other words, we define $E[T]=\{\{u_i, u_{i+1}\}\mid i=1,2,...,k-1\}\cup \{\{u_k, u_1\}\}$ and $V[T]=\{u_i\mid i=1,2,...,k\}$.

We are given a complete graph $G=(V,E)$, a root vertex $r\in V$, and an integer $k$. Let $c:E\to \R_+$ be a nonnegative cost function on $E$ and $\pi:V\to \R_+$ be a nonnegative penalty function. The $k$-prize collecting traveling salesman problem ($k$-PCTSP) is a problem of finding a tour $T$ that minimizes $\sum_{e\in E[T]}c(e)+\sum_{v\notin V[F]}\pi(v)$ in the tour of G with the number of vertices equal to or larger than $k$ and including the vertex $r$. This problem is a generalization of both $k$-traveling salesman problems($k$-TSP) and penalty traveling salesman problems(PTSP)\cite{Han2016}. $k$-TSP and PTSP is described in \cite{Handbook}. Under the triangular inequality for the edge cost function, Han et al. show that there is a 5-approximation algorithm for the $k$-PCTSP in \cite{Han2016}, which is the best approximation ratio among the approximation algorithms for the existing $k$-PCTSP.

It is known that $k$-PCTSP is a special case of prize collecting traveling salesman problems(PCTSP) proposed by Balas\cite{Balas}. Given a complete graph $G=(V,E)$, a vertex $r\in V$, a nonnegative integer $Q\in \Z_+$, a cost function $c: E \to \mathbb{R}_+$, a penalty function $\pi: V \to \mathbb{R}_+$, and a reward function $w: V \to \mathbb{Z}_+$. This problem finds a tour $T$ such that contains $r$ and minimizes $\sum_{e \in E[T]}c(e) + \sum_{v \not\in V[T]}\pi(v)$ under a cover constraint $\sum_{v \in V[T]}w(v) \geq Q$. PCTSP can be reduced to $k$-PCTSP by setting $w(v)=1$ for every $v\in V$ and $Q=k$. 

An approximation algorithm combining two approximation algorithms for quota traveling salesman problems (QTSP) and for PTSP is shown in \cite{Awerbuch} and \cite{Handbook} for the PCTSP assuming that the cost function satisfies the triangular inequality. From the fact that QTSP is a generalization of $k$-TSP, we can show a 4-approximation algorithm for $k$-PCTSP when the cost function satisfies the triangular inequality by using a similar technique.

Our algorithm framework can be applied to the $k$-PCTSP. We employ two 2-approximation algorithms, one is for PTSP\cite{G-W} and the other is for $k$-TSP\cite{Garg2005}.

\begin{itembox}[l]{4-approximation algorithm for $k$-PCTSP}
\begin{description}
\item [Input: ] A complete graph $G = (V, E)$, a vertex $r \in V$, an integer $k>0$, a cost function $c: E \to \mathbb{R}_+$, and a penalty function $\pi: V\to \R_+$.
\item [Output: ] A tour of $G$ including $r$, denoted by $T_{\rm OUT}=(V_{\rm OUT}, E_{\rm OUT})$.

\item[Step 1:] For a $k$-TSP instance $(G, r, k,c)$, we apply Garg's algorithm. Let $T_{\rm Garg}$ be an output of this procedure.
\item[Step2: ] For a PCTSP instance $(G, r, c, \pi)$, we apply GW algorithm. Let $T_{\rm GW}$ be an output of this procedure. 
\item[Step3: ] Return a short cut of a tour that merged $T_{Garg}$ and $T_{GW}$.
\end{description}
\end{itembox}

From each procedure is 2-approximation algorithm for $k$-TSP and PTSP, we can get the following by using \cite{Awerbuch}'s technique.

\begin{theorem}
Proposed algorithm is 4-approximation algorithm for $k$-PCTSP under the triangular inequality.
\end{theorem}
\begin{proof}
Suppose that $p$, $p_{\rm Garg}$, and $p_{\rm GW}$ are costs of $T_{\rm OUT}$, $T_{\rm Garg}$, and $T_{\rm GW}$, respectively. We denote the optimal value for each problem to OPT$_{k{\rm -TSP}}$
, OPT$_{\rm PTSP}$, and OPT.

$$p \leq p_{\rm Garg} + p_{\rm GW} \leq 2\mathrm{OPT_{k{\rm -TSP}}} + 2\mathrm{OPT_{\rm PTSP}} \leq 4\mathrm{OPT}.$$
\end{proof}

\section{Conclusion}
This paper proposed 4-approximation algorithm for both $k$-prize collecting Steiner tree problems and $k$-prize collecting traveling salesman problems. 
The time complexity of the proposed approximation algorithm is O ($mn^4 \log n$). A bottleneck is Garg's approximation algorithm for $k$-MST, $k$-TSP. As a future task, in order to reduce the calculation amount, reduction of the time complexity of the approximation algorithm for $k$-MST and $k$-TSP can be mentioned. Improvement of the approximation rate for $k$-PCST and $k$-PCTSP is also a future task.

\section*{Acknowledgement}
This research was partially supported by the Ministry of Education, Science, Sports and Culture through Grants-in-Aid for Scientific Research (C) 26330025, (B) 15H02972, and through Grants-in-Aid for Young Scientists (B) 26870200.
\\

\end{document}